\documentclass[10pt,a4paper]{article}

\usepackage{graphicx}
\usepackage{hyperref} 
\usepackage{amsmath,amsfonts,amsthm,amssymb}%
\usepackage{tabularx}
\usepackage[shortlabels]{enumitem}
\usepackage{xspace}
\usepackage{authblk}
\usepackage{multirow}
\usepackage{longtable}
\usepackage{graphicx}
\usepackage[colorinlistoftodos]{todonotes}

\def\E{\mathcal{E}}
\def\V{\mathcal{V}}
\def\W{\mathcal{W}}
\def\L{\mathcal{L}}
\def\C{\mathcal{C}}

\newtheorem{theorem}{Theorem}[section]
\newtheorem{proposition}[theorem]{Proposition}

\newtheorem{property}{Property}

\title{A graph complexity measure based on the spectral analysis of the Laplace operator}

\author[1,2,3*]{Diego M. Mateos }
\author[3] {Federico Morana}
\author[1,3]{Hugo Aimar}

\affil[1]{\normalsize Consejo Nacional de Investigaciones Cient\'ificas y T\'ecnicas (CONICET), Argentina.}
\affil[2]{\normalsize Facultad de Ciencia y Tecnolog\'{\i}a. Universidad Aut\'{o}noma de Entre R\'{\i}os (UADER). Oro Verde, Entre R\'{\i}os, Argentina.}
\affil[3]{\normalsize Instituto de Matem\'{a}tica Aplicada del Litoral (IMAL), UNL, CONICET, CCT CONICET, Santa F\'e, Argentina.}

\affil[*]{Corresponding author: Diego M. Mateos, mateosdiego@gmail.com.}
\date{}

\begin{document}

\maketitle

\begin{abstract}
In this work we introduce a concept of complexity for undirected graphs in terms of the spectral analysis of the Laplacian operator defined by the incidence matrix of the graph. Precisely, we compute the norm of the vector of eigenvalues of both the graph and its complement and take their product. Doing so, we obtain a quantity that satisfies two basic properties that are the expected for a measure of complexity. First, complexity of fully connected and fully disconnected graphs vanish. Second, complexity of complementary graphs coincide. This notion of complexity allows us to distinguish different kinds of graphs by placing them in a ``croissant-shaped'' region of the plane link density - complexity, highlighting some features like connectivity, concentration, uniformity or regularity and existence of clique-like clusters. Indeed, considering graphs with a fixed number of nodes, by plotting the link density versus the complexity we find that graphs generated by different methods take place at different regions of the plane. We consider some of the paradigmatic randomly generated graphs, in particular the Erd\"os-R\'enyi, the Watts-Strogatz and the Barab\'asi-Albert models. Also, we place some particular, let us say deterministic, well known hand-crafted graphs, to wit, lattices, stars, hyper-concentrated and cliques-containing graphs. It is worthy noticing that these deterministic classical models of graphs depict the boundary of the croissant-shaped region. Finally, as an application to graphs generated by real measurements, we consider the brain connectivity graphs from two epileptic patients obtained from magnetoencephalography (MEG) recording, both in a baseline period and in ictal periods (epileptic seizures). In this case, our definition of complexity could be used as a tool for discerning between states, by the analysis of differences at distinct frequencies of the MEG recording.
\end{abstract}

\section{Introduction}
\label{sec:Introduction}
The Laplace operator in the Euclidean setting provides the Hamiltonian of the free particle in quantum mechanics. On the other hand, the classical Laplace operator is the Euler-Lagrange equation provided, in a domain $\Omega$ of the space, by $\L(\varphi)=\int |\nabla \varphi|^2 $. For a finite set $\{ 1,2,...,n \}=\mathcal{V}$ and functions $\varphi$ taking real or complex values defined in $\V$ a discrete analogous of $\L$ is $\mathcal{L}(\varphi)=\sum_{i=1}^n \sum_{j=1}^n |\varphi_i - \varphi_j|^2$. Now, the corresponding Euler-Lagrange equation is $\sum_{i=1}^n (f_i - f_j)=0$ for every $i \in \mathcal{V}$. When instead of the whole set $ \mathcal{V} \times \mathcal{V}$ we have an undirected graph, with vertices in $\mathcal{V}$ and adjacency matrix $\mathcal{W}=(w_{ij})$ taking the value 1 when there is an edge joining vertices $i$ and $j$ and zero otherwise, the natural Lagrange functional is $\L_\W (\varphi)=\sum_{i=1}^n \sum_{j=1}^n w_{ij}|\varphi_i - \varphi_j|^2$. The corresponding Euler - Lagrange equation is provided by the graph Laplace equation $\Delta \varphi=0$, with $\Delta \varphi(i)=\sum_{j=1}^n w_{ij}(\varphi_j - \varphi_i)$, $w_{ii}=0$ and $w_{ij}=w_{ji}$. 
The above considerations suggest that the spectral theory of this discrete Laplace operator will provide the fundamental states and the corresponding energy levels for a ``free particle'' in the graph defined by $\W$ and $\V$. The energy levels are the eigenvalues of $\Delta$. In the classical setting a well known  measure of information is given by the Von Neumann entropy that takes into account the eigenvalues of the Laplacian. Here, we aim to use the spectral analysis of the discrete Laplacian  in order to introduce a notion of complexity for graphs. We shall precisely define it in the next section. Actually we shall determine a region in the plane of the variables \textit{link density-complexity} which contains the most well known random graphs such as the Erd\"os-R\'enyi, the Watts-Strogatz and the Barab\'asi-Albert, but also the deterministic (non-random) classical graphs. As an application to real data we consider the brain connectivity graphs for different states of epileptics patients through magnetoencephalogaphy (MEG). Section \ref{sec:Mathmodel} describe briefly the mathematical model, some proofs are contained in Appendix \ref{sec:Appendix}. Section \ref{sec:BuildPlane} is devoted to empirically describe the croissant-shape of the region in the plane of the variables \textit{link density-complexity}, which we conjecture  contains all the graphs, random or not. Section \ref{sec:RandomModels} contains the analysis and the result for the three most paradigmatic random graphs: Erd\"os-R\'enyi, Watts-Strogatz and Barab\'asi-Albert. In Section \ref{sec:ApplicationNeuro} we provide application to neurophysiological data. Section \ref{sec:Discussion} contains comments and conclusions. As we said before Appendix \ref{sec:Appendix} describes  some of the mathematical proofs.

\section{The mathematical model}
\label{sec:Mathmodel}
Let $G=\left(\mathcal{V},\mathcal{E},\W\right)$ be a simple undirected graph, where $\mathcal{V}=\{1,\dots,n\}$ is the set of vertices or nodes, $\mathcal{E}=\{e_1,\dots,e_m\}\subset \{ \{i,j\}\colon i,j\in \mathcal{V}\}$ is the set of edges and $\W\colon \mathcal{V}\times\mathcal{V}\to\{0,1\}$ is the adjacency matrix of $G$ with $w_{ij}=1$ whenever $\{i,j\}\in \mathcal{E}$ and zero otherwise. Since the graph is undirected and simple the matrix $\W$ is symmetric with null diagonal. We will denote $i\sim j$ when $\{i,j\}\in \mathcal{E}$.

The degree of a vertex $j$ is defined by $\delta(j)=\sum_{i\in\mathcal{V}} w_{ij}$. The degree matrix is defined as the diagonal $n\times n$ matrix containing the degrees of the nodes and denoted by $D=diag(\delta(1),\dots,\delta(n))$. The Laplacian of the graph is the lineal operator acting on real or complex functions defined on the nodes, with matrix given by 
\begin{equation}\label{eq:laplacian}
    \Delta=\W-D.
\end{equation}
This operator is symmetric and negative semi-definite. Therefore we can apply the spectral theorem to obtain an orthonormal basis of $\ell^2(\mathcal{V})\sim\mathbb{R}^n$ of eigenvectors $\{\psi_1,\dots,\psi_n\}$ of $\Delta$. It is usually called the Fourier basis of $G$. The associated eigenvalues $\{\lambda_1,\dots,\lambda_n\}$ satisfies $0=\lambda_1\geq\lambda_2\geq\dots\geq\lambda_n$. 
In the following we will refer to the vector $\overline{\lambda} = (\lambda_1,\dots,\lambda_n)$ as the spectrum of the graph $G$. The trace of the Laplacian is a feature of interest for our further analysis and is given by $\sum_{i=1}^n \lambda_i=-2m$, where $m$ is the number of edges of $G$. For a general reference regarding the spectral theory of the Laplacian on graphs see \cite{BroBruLeCSzlVan17} and references therein.

With the energy point of view described in Section \ref{sec:Introduction}, we may consider equivalent two graphs $G$ and $G'$ that share the spectrum $\bar{\lambda}$. Hence, for $G$ and $H$ two graphs with the same number $n$ of vertices, the function $d_s(G,H)=\left \| \bar{\lambda}_G- \bar{\lambda}_H \right \|$, with 
$\bar{\lambda}_G$ and  $\bar{\lambda}_H$ the spectral vectors of $G$ and $H$ respectively and $\left \| \cdot \right \|$ any norm in $\mathbb{R}^n$, is a distance (metric) between the classes of co-spectrality of $G$ and $H$. We shall take $\left \| \bar{\lambda } \right \|=\left ( \sum_{i=1}^{n} |\lambda_{i}|^2 \right )^{1/2}$ the usual norm. We shall refer to $d_s$ as the spectral distance. Notice that since the first eigenvalue $\lambda_1$ of each graph vanishes, we actually have that $d_s (G,H)=| \bar{\Lambda}_G - \bar{\Lambda}_H  |$, where $ \bar{\Lambda }=\left ( \lambda_2,...,\lambda_n \right )$ and $ | \cdot | $ is the euclidean norm in $\mathbb{R}^{n-1}$. 
The spectral distance on graphs was considered before in \cite{gu2015spectral}, see also \cite{deza2016encyclopedia}. 

In order to introduce our definition of spectral complexity of a graph, let us set $Z$ to denote the null graph, i.e $w_{ij}=0$ for every $i,j=1,...,n$ and $F$ the complete graph i.e $w_{ij}=1$ for every $i \neq j$. Now we can define the \textbf{spectral complexity} of a graph $G$ with $n$-vertices 

\begin{equation}
    \begin{split}
        \C_s(G) & = d_s(G,Z) \cdot d_s(G,F) \\
        & = \left \| \bar{\lambda}_G - \bar{\lambda}_Z \right \|~\left \| \bar{\lambda}_G - \bar{\lambda}_F \right \| \\
        &=| \bar{\Lambda}_G - \bar{\Lambda}_Z |~| \bar{\Lambda}_G - \bar{\Lambda}_F |.
\end{split}
\end{equation}

Two basic premises are behind this definition. The first one is that both, the null graph and the full graph are the less complex graphs that can be defined on the vertices set $\V=\{ 1,...,n\}$. The second is that complementary graphs should have the same complexity. The following properties that we prove in Appendix \ref{sec:Appendix} show that our definition of spectral complexity $\C_s$ satisfies those two requirements. 

\begin{enumerate}[a)]
\item If $F$ is the complete graph, then $\bar{\Lambda}_F=\bar{n}=\left ( n, n, ...,n \right ) \in \mathbb{R}^{n-1}$.
\item If $Z$ is the null graph, then $\bar{\Lambda}_Z=\bar{0}=\left ( 0, 0, ...,0 \right ) \in \mathbb{R}^{n-1}$.
\end{enumerate}

If $w_{ij}$ is the adjacency matrix of $G$, the complement $G^c$, of $G$ is the graph defined by $w_{ij}^c=1$ if $w_{ij}=0$, $i\neq j$, and $w_{ij}^c=0$ if $w_{ij}=1$. Let $\bar{\lambda}=\left ( \lambda _1,...,\lambda_n \right )$ denote the spectral vector of $G$ and let $\bar{\lambda}^c$ denote the spectral vector of $G^c$. Then, 

\begin{enumerate}[c)]
\item $\bar{\Lambda}^c=-\bar{n}-\bar{\Lambda}$.
\end{enumerate}

These basic facts provide a directly computable formula for the above defined spectral complexity,

\begin{equation}
    \C_s(G)=| \bar{\Lambda}  | \cdot | \bar{\Lambda}^c  |=| \bar{\Lambda}  | \cdot | \bar{n} + \bar{\Lambda}  |.
\end{equation}

A second quantity associated to a graph that we shall take into account in our analysis is its link density. The link density $\rho$ of a simple unidirected graph is the number of actual edges divided by the number of all possible edges. With our notation 

\begin{equation}
    \rho(G)=\frac{2m}{n(n-1)}.
\end{equation}

Given a positive integer $n$ we shall display all the possible graphs $G$ built on $\V=\{ 1,2,...,n \}$ in the plane of the variables $\rho(G)$ and $\C_s(G)$. Since the density of a graph $G$ and the density of its complement can be quite different, actually $\rho(G)+\rho(G^c)=1$, it is clear that  the link density is not a function of the spectral complexity. It is also simple to show that graphs with the same density may have different spectral complexity. So neither $\rho$ is a function of $\C_s$ nor $\C_s$ is a function of $\rho$. As could be expected. Nevertheless $\rho$ and $\C_s$ are not completely independent. In fact we empirically determine the region in the region in the plane ($\rho$, $\C_s$) spanned by all possible graphs. It is clear from its very definition that no matter how large is $n$, the link density is normalized $0 \leq \rho \leq 1$. This is not the case for $\C_s$ as defined above. Hence in order to be able to compare the values of $\C_s$ for graph with different numbers of vertices, we shall also empirically normalize the spectral complexity. So that the region that we are looking for will be a subset of the unit square $\left [ 0,1 \right ] \times  \left [ 0,1 \right ]$.

\bigskip
\section{The ``croissant-shaped'' domain for the pairs ($\rho$, $\C_s$) for every $n$ }
\label{sec:BuildPlane}

The delimitation of the region in the representation plane \emph{link density-complexity} where all the variety of graphs take place is not a trivial task to perform theoretically. Here we obtain an empirical approximation of the upper and lower boundaries, derived by placing a wide variety of graphs generated by random and deterministic methods.

Having fixed the number of nodes, we observe that the most complex graphs of a given number of edges are those where the connections are concentrated in few nodes (usually called `hubs') and, conversely, the graphs with node degrees uniformly distributed are the less complex ones. The former are named `multi-stars' or `multi-hubs' type graphs and the latter are the so called `lattices'. In between fall many of the paradigmatic well-known graphs, each found in some particular point regarding its nature. As well, taking into account the second principle of our definition of complexity, this behaviour gets reflected in the complementary graphs. In that regard, we shall notice that the complement of a multi-stars graph is a graph with one clique and a remaining number of isolated nodes, in correspondence with the number of hubs of the original graph. On the other hand, the complement of a lattice type graph is another lattice. As the extreme complexity values are obtained for the most pure form of the described type of graphs, those are the ones that give us the boundaries.

Let us first consider graphs of a fixed number of nodes $n$. Here the upper limit is constructed simply as the polygonal joining the points corresponding to the multi-stars graphs of link density less or equal than $0.5$, and the points corresponding to their complements to complete the side of the graphs with link density greater than $0.5$. In particular, the multi-stars graph with lower link-density  is the single-star graph, which is the graph with one node connected to every other node and no extra links. The following multi-star graph shall be the one with two fully connected nodes and no extra attachments. And so on while the link density remains lower than a half. In the same manner the lower limit consists on the polygonal joining the points corresponding to the regular lattices. For instance, we shall find there the placement of the cycle graph. These upper and lower boundaries form a rough croissant-shape contained in the unit square $\left [ 0,1 \right ] \times  \left [ 0,1 \right ]$. This region is similar for every $n$, although slightly increasing with $n$, it seems to stabilize asymptotically when $n$ tends to infinity.
Then we can go further ahead and consider the limit shape.
That will be the domain for the family of all the graphs of any finite number of nodes.



The boundary graphs can be generated with a recursive algorithm that at each step adds some particular edges to the previous graph, starting with the null graph and ending with the complete graph. Essentially we are generating a sequence of adjacency matrices by replacing zeroes by ones at the positions that correspond to the edges being added. For example, if we add an edge between the vertices $i$ and $j$ then the new adjacency matrix will have a one at the positions $(i,j)$ and $(j,i)$. In account of the symmetry, we will show how the algorithms acts on the upper triangular part of the adjacency matrices of the sequences of graphs that are being generated. 


In order to generate the upper boundary we start with the triangular array of zeros and we fill it with rows of ones at steps, from up to down, until completion (Figure~\ref{fig:matrices}.A).
%
%
In order to obtain the lower bound we do as before but now moving diagonally, in such a way that a ring is added to the graph at each step. 
Precisely, we fill the following sequence of pairs of diagonals in order $(1, n-1)\to (2, n-2)\to (3, n-3)\to \dots$ 
until the all-ones array is reached, as shown in Figure~\ref{fig:matrices}.B.

\begin{figure}[h]
\centerline{\includegraphics[width=\textwidth]{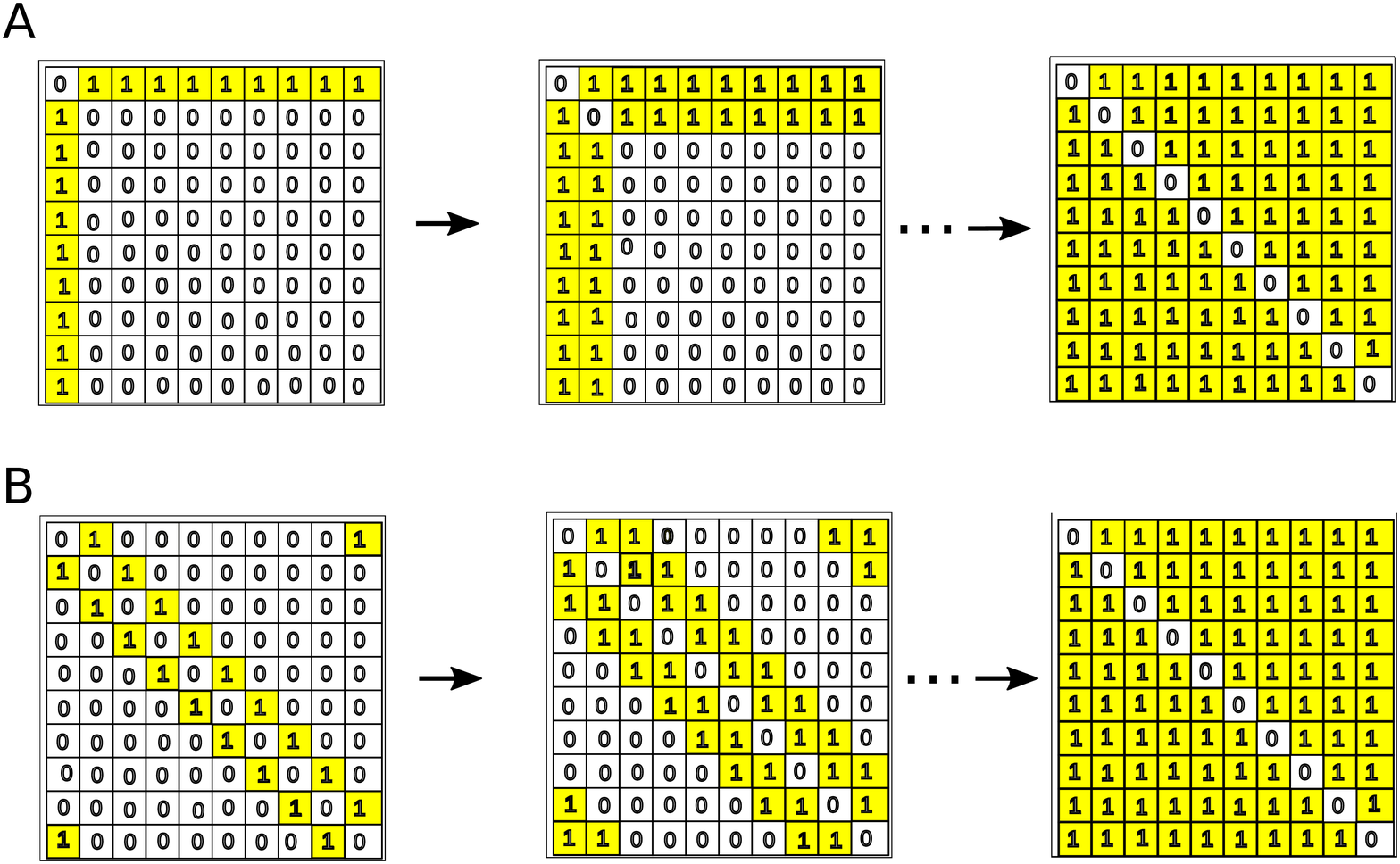}}
\caption{ Generation of the croissant-shaped boundary. A) The sequence of adjacency matrices of the graphs that 
configure the upper limit is generated recursively, 
starting with a matrix of zeros and filling its upper triangular part with a row of ones at each step, 
from top to bottom until completion (the lower triangular part of each matrix is completed by symmetry). B) For the lower boundary we do as before but the ones are placed in such a way that a ring is added to the graph at each step.}
\label{fig:matrices}
\end{figure}

Figure \ref{fig:mapa_completo} depicts the croissant-shaped region for $N=100$ and the placement of some paradigmatic graphs of 15 vertices.

\begin{figure}[h]
\centerline{\includegraphics[width=\textwidth]{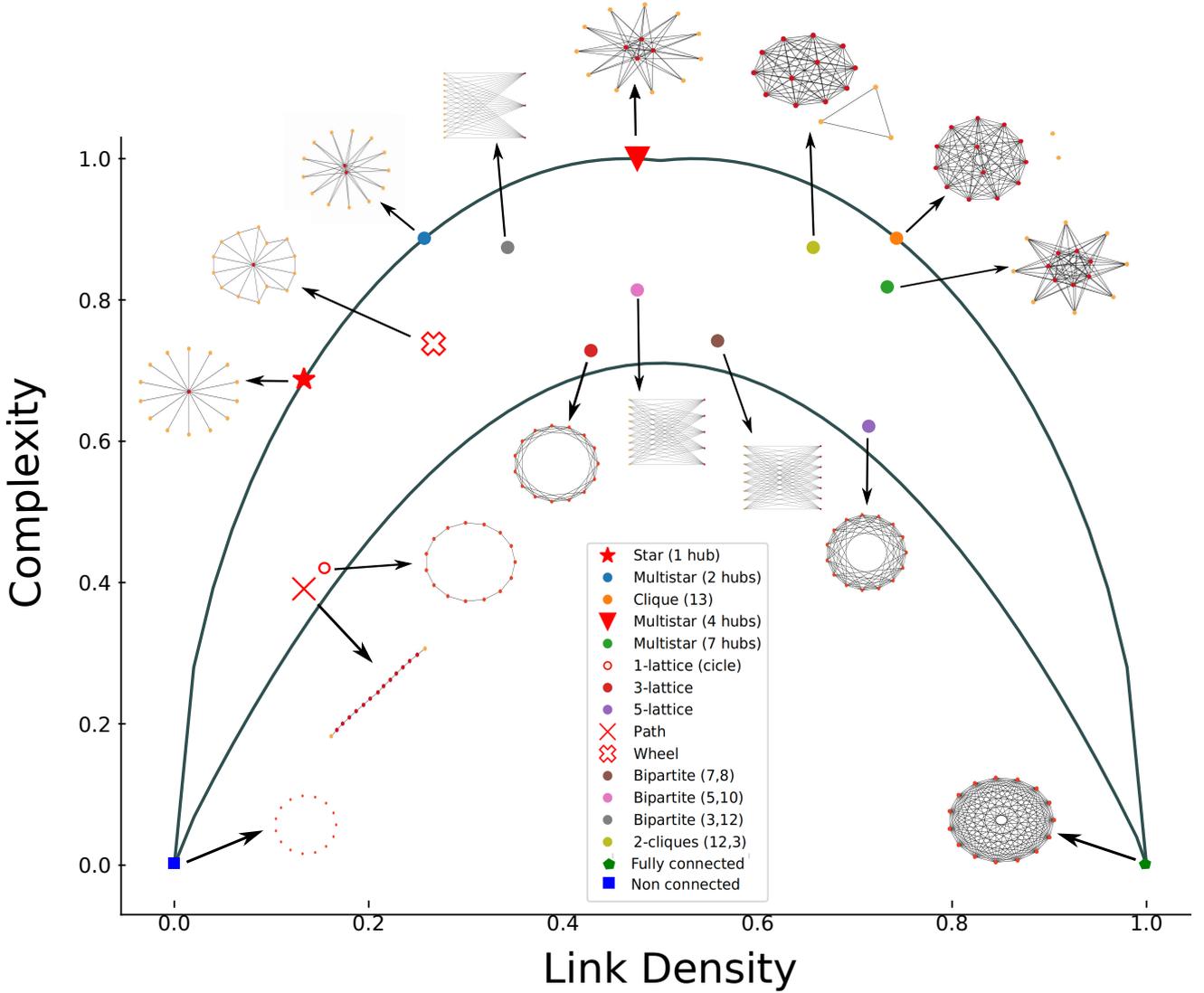}}
\caption{  Schematic distribution of the different types of networks for $n=15$ on the ``croissant-shaped'' region.}
\label{fig:mapa_completo}
\end{figure}

\section{Random graphs}
\label{sec:RandomModels}

In this section we consider and plot for several values of $n$, three well known stochastic models, the Erd\"os-R\'enyi model, the Watts-Strogatz model, and the Barab\'asi-Albert model. As we shall see, each 
of them draws some characteristic pattern contained in 
the basic croissant shape.
In particular, the Erd\"os-R\'enyi model closely depicts the lower bound curve of the region.

\subsection{Erd\"os-R\'enyi model}
The Erd\"os-R\'enyi (ER) is one of the simplest models for generating random graphs (\cite{gilbert1959random}, \cite{erdHos1960evolution}).
We shall use the approach introduced by Edgar Gilbert in 1959, the so called $G(n,p)$ model. In this model a graph with $n$ vertices is constructed by connecting nodes randomly, including each edge in the graph with probability $p$ independently from every other edge. 
The parameter ``\textit{linking probability}'' $p$ is the expected value of the link density of a $G(n,p)$ generated graph.

For the analysis of the model we consider graphs with number of nodes $n=100, 200, 300, 400$ and linking probability $p$ for values equispaced in $[0,1]$ with $\Delta p=0.01$. For each fixed pair of parameters $(n,p)$ we generate $100$ graphs and compute their complexity and link density. On these sets of values we calculate the mean values and standard deviations, and plot the resulting point and deviation in the link density vs complexity plane. The result is shown in Figure~\ref{fig:ER_model}.

\begin{figure}[h]
\centerline{\includegraphics[width=\textwidth]{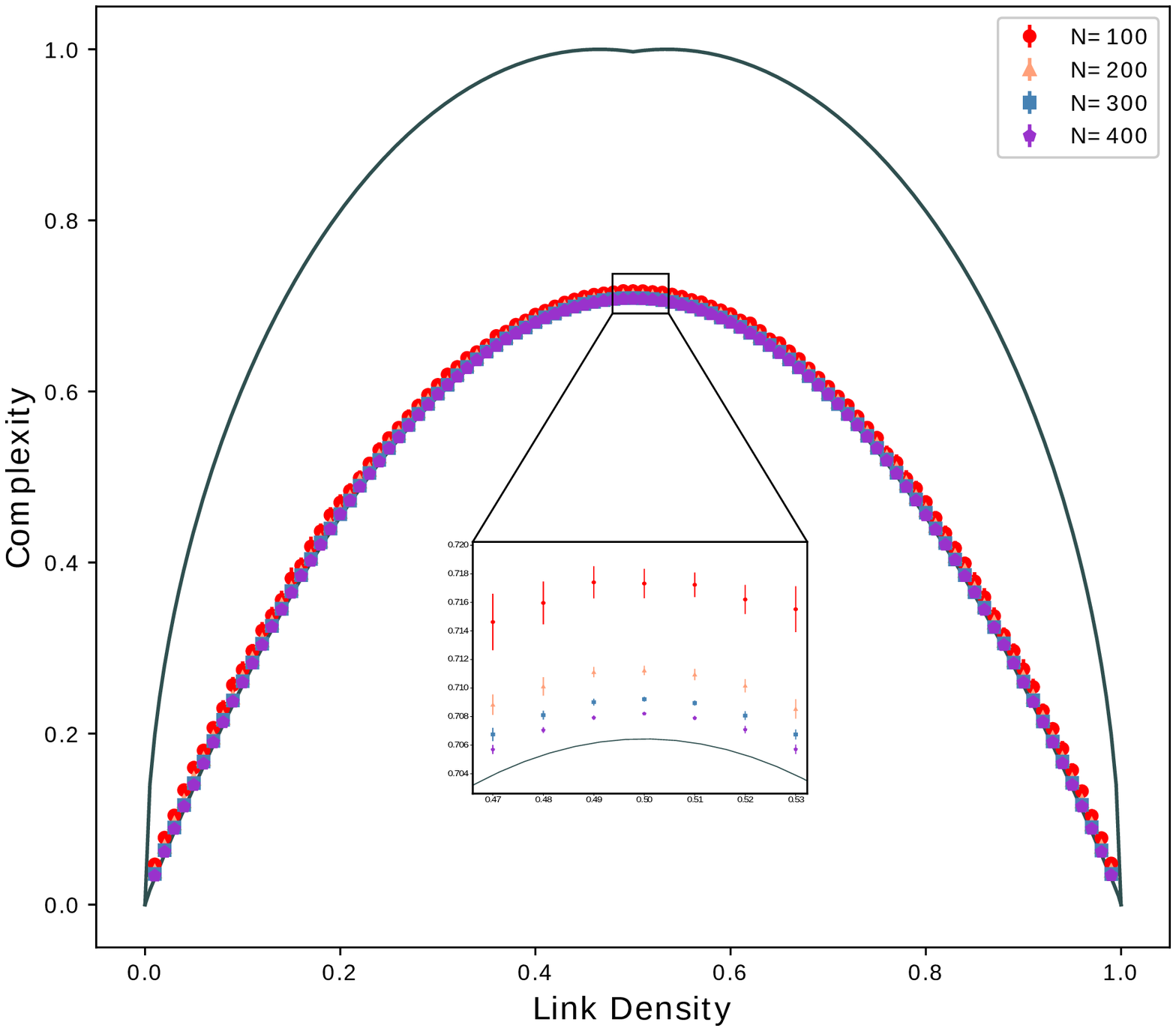}}
\caption{  Analysis of the Erd\"os--R\'enyi model for different linkage probability ($p  \in [0,1] $) and number of nodes  $n=100,200,300,400$. Each point/bar represents the mean/standard deviation performed over 100 iterations.}
\label{fig:ER_model}
\end{figure}

As we have said before this model closely follows the lower bound curve of the croissant shaped region. Thereby, as could be expected, the maximum complexity occurs at $p=0.5$, for every $n$.

\subsection{Watts-Strogatz model }
The Watts-Strogatz model (WS) is a graph generation method based on rewiring edges randomly. It starts with a $k$-ring (or lattice) graph of $n$ nodes, with $1\leq k\leq\frac{n-1}{2}$, and a rewiring probability $\beta$.
In a $k$-ring, each node is connected with $k$ neighbours at each side (left and right). 
The algorithm consist in going throughout the nodes, one by one, and at each node to put in consideration the rewiring of the edges connected to a right neighbour. Each of these edges will be rewired with a probability $\beta$ to form a new connection of the present node to another randomly selected between the nodes that are not currently connected to it.  This model can be though of as an interpolation between a lattice and something close to an ER graph. 
When $\beta=0$ there is not rewiring and the initial ring lattice is preserved. On the other hand, for $\beta=1$ the whole structure is reconfigured randomly. To some extent this extreme case resembles the ER model, but not quite, since every node will surely remain connected to at least $k$ other nodes. For intermediate values of $\beta$, the model generates networks with the so called \textit{small-world} property, including short average path lengths and high clustering (\cite{watts1998collective}, see also \cite{travers1977experimental} and \cite{boccaletti2006complex}).

Again, with this method we generate graphs with number of nodes $n=100, 200, 300, 400$, number of rings $k$ ranging from $1$ to $\frac{n-1}{2}-\frac{1}{2}$ (since we take $n$ even) and the parameter $\beta$ varying from $0$ to $1$ with $\Delta \beta=0.1$. For each triplet $(n,k,\beta)$ a hundred graphs were generated calculating their complexities and the mean value and standard deviation of that set of values. Notice that the link density of a $(n,k,\beta)$ WS--graph is given by $\frac{2k}{n-1}$. Then those points were plotted in the link density vs complexity plane, as shown in Figure~\ref{fig:WS_model}.

The experiments show that the WS--graphs with $n$ and $k$ fixed (having common link density) have bigger complexities as $\beta$ increases, with a minimum when $\beta=0$ (coincident to the corresponding lattice complexity) and where an upper bound is given by the expected value of the complexities of the ER--graphs with the same link density.
This is consistent with the previous observation that this model is a sort of 
interpolation between a lattice and an ER graph.

\begin{figure}[h]
\centerline{\includegraphics[width=\textwidth]{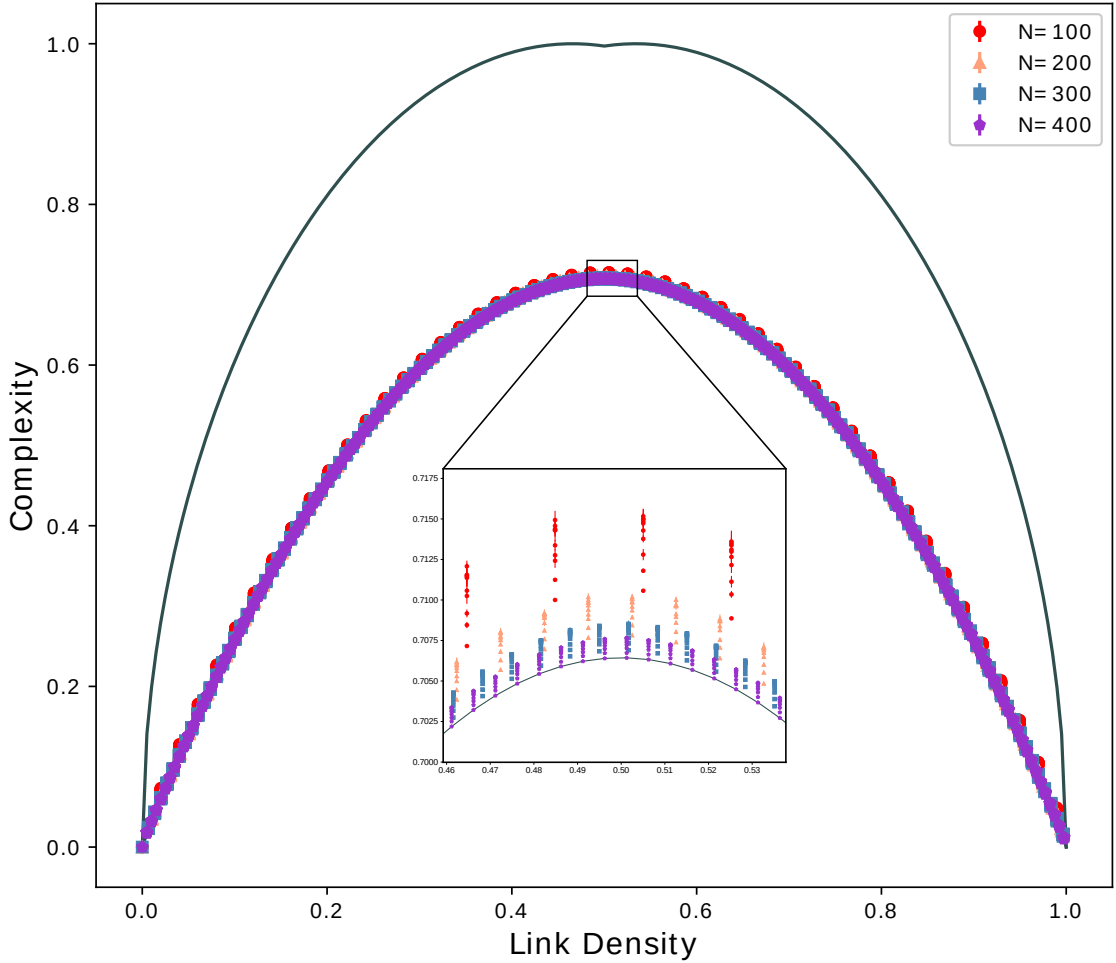}}
\caption{ Analysis of the Watts--Strogatz model for parameters $k \in [1, \frac{n}{2}-1]$ and $\beta \in [0,1] $ with $\Delta \beta=0.1$, and number of nodes $n=100,200,300,400$. Each point/bar represents the mean value/standard deviation performed over 100 iterations.}
\label{fig:WS_model}
\end{figure}

\subsection{Barab\'asi-Albert model}
The Barab\'asi--Albert (BA) model produces graphs by a random growing mechanism (\cite{barabasi1999emergence}). It is determined by an integer parameter $m$, that has to be positive and smaller than the number of nodes $n$. The process begins with a completely disconnected graph of $m$ nodes that evolves in stages, adding a node and $m$ edges attached to it at each stage, until a graph of $n$ nodes is obtained. 
The edges that connect every new node are decided by preferential attachment, which is a probabilistic method of picking the nodes to be attached with a probability proportional to the actual degree of the existing nodes. 
Both growth and preferential attachment exist largely in real networks, so the graphs obtained by this method actually share some properties with them. 
In particular, preferential attachment produces hubs (highly connected nodes) and peripheral communities, where nodes have similar degree. The hubs are few but with much higher degree. 
Actually, the degree distribution of the nodes can be fitted by a power--law $P(i) \sim i^{-\gamma}$, where $\gamma$ is the degree exponent, usually $2 \leq \gamma \leq 3$. That is the so called \textit{scale--free} property, commonly observed in many social networks.

For this case we generate graphs with number of nodes $n=100, 200, 300, 400$ and parameter $\ell$ varying between $1$ and $n$. For each fixed values of $n$ and $\ell$, $100$ graphs were generated and with the same process as in the previous models we obtain Figure~\ref{fig:BA_model}.

The pattern depicted in the plane by the BA--graphs is the most peculiar of the three models considered. It is worthy to notice first that the shape of the curve determined when the parameter $\ell$ runs through its entire domain remains the same regardless of the number of nodes, moreover it stabilizes for large values of $n$. 
In the following observations we consider BA--graphs of a fixed number of nodes $n$. 
The link density of the graphs grows together with the parameter $\ell$ until the critical value $\ell_1=\frac{n}{2}$ is reached, where the link density gets
slightly over 0.5, and then it decreases. However, the complexity continues to grow for a while up to another threshold at some point $\ell_2>\ell_1$. 
Furthermore, the complexity gap between 
the curve and the lower boundary of the croissant shaped region keeps increasing for higher values of $\ell$. 
As well, when $\ell$ gets closer to $n$ the complexity values approach to the upper frontier of the croissant shaped region, making contact 
at $\ell=n-1$.
So we have that at the same link density the method produces two graphs of quite a different type in relation to their complexity. 
That could be explained by the greater concentration (fewer and more connected hubs) present in the graphs obtained at higher values of $\ell$.

\begin{figure}[h]
\centerline{\includegraphics[width=\textwidth]{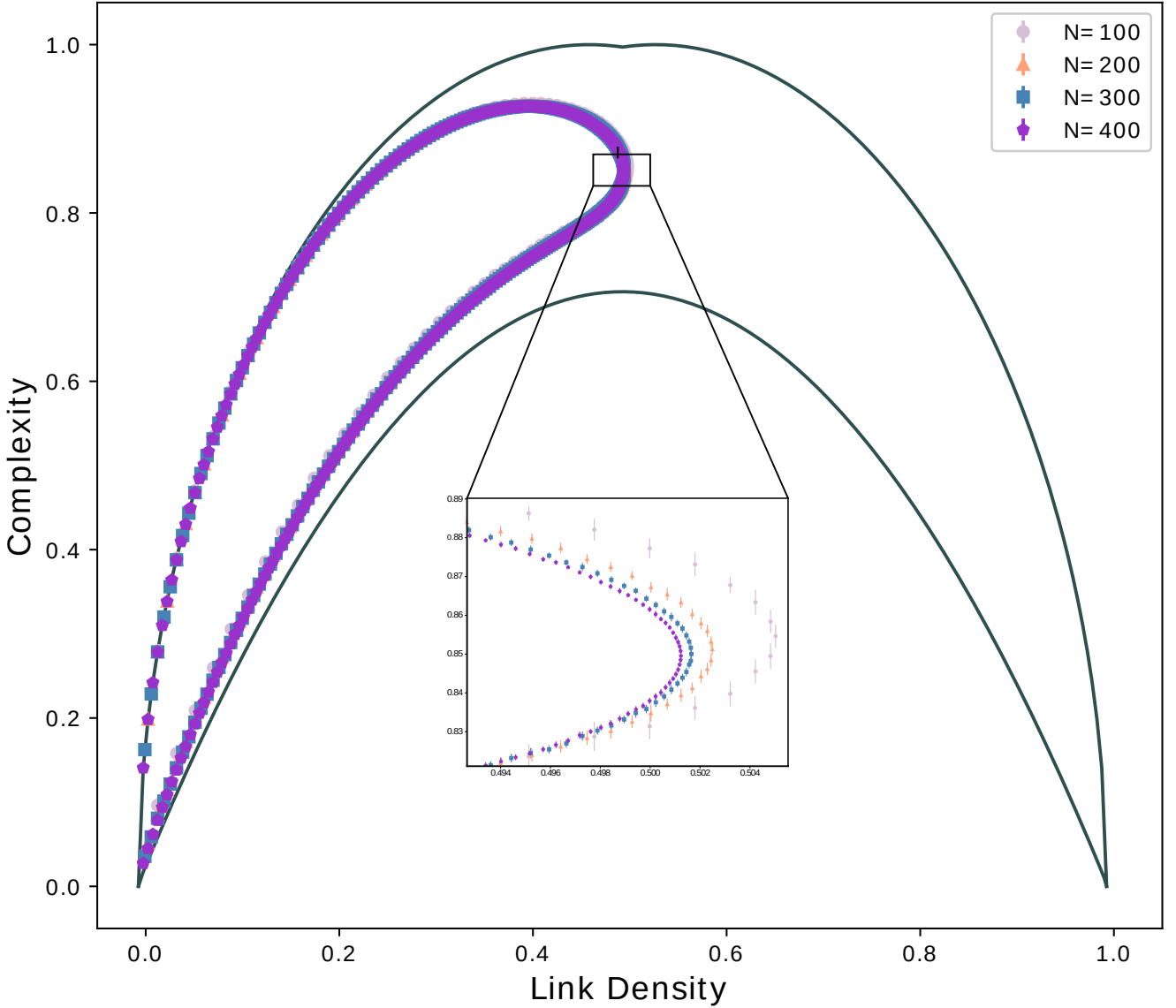}}
\caption{ Analysis of the Barab\'asi-Albert model for graphs with different number of nodes $n=100,200,300,400$ and parameter of preferential attachment $\ell$ ranging from $1$ to $n$. Each point/bar represents the mean value/standard deviation performed over 100 iterations.}
\label{fig:BA_model}
\end{figure}

\begin{figure}[h]
\centerline{\includegraphics[width=\textwidth]{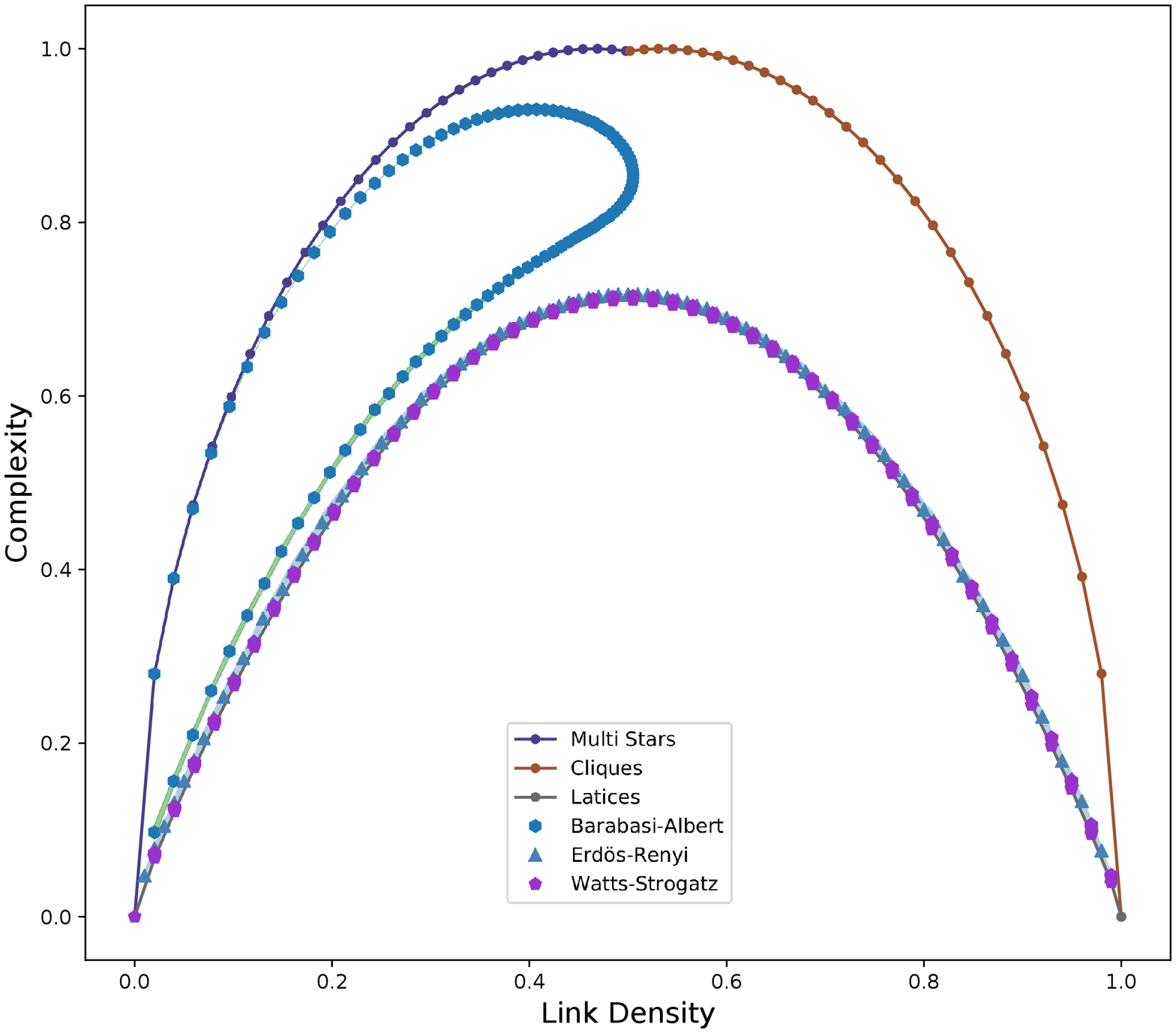}}
\caption{  Overview and comparison of the results obtained for all the network models discussed before. In this case we use for all models $n=100$. }
\label{fig:Allmodels}
\end{figure}

\section{Application to neurophysiological data}
\label{sec:ApplicationNeuro}

To test the complexity on real data, we analyse brain connectivity graphs over two epileptic patients. The former, ``patient 1''  suffering primary generalized epilepsy, and the latter ``patient 2'' with secondary generalized epilepsy. Details of the patients’ epilepsy's, seizure types, and the recording specification have been presented in previous studies \cite{dominguez2005enhanced}. Each patient underwent a magnetoencephalography (MEG) recording in the Baseline period  and in the ictal period (Seizure). The recording was performed using a 144 channel MEG with 625 Hz sampling rate. In each state, we took 17 windows of  5-second recording and we calculated the functional connectivity. The connectivity  was implemented through the Phase Synchrony Index ($R$) method over all possible pairwise signal combinations. The methods have been extensively described in several publications, so we refer the reader to a few representative papers \cite{dominguez2005enhanced,mormann2000mean}.

The $R$ analysis was performed for five central frequencies $f= (3, 5, 10, 20, 30) \pm 2 Hz$. 
This generated $17 \times 5$ connectivity matrices for the baseline and seizure state. The connectivity matrices were binarised using a threshold defined by a surrogate analysis (see \cite{mateos2017consciousness}).  Complexity and link density was calculated for all matrices. We computed the mean value and standard deviation of both quantities and plotted in the \textit{complexity-link density} plane. 

\begin{figure}[h]
\centerline{\includegraphics[width=\textwidth]{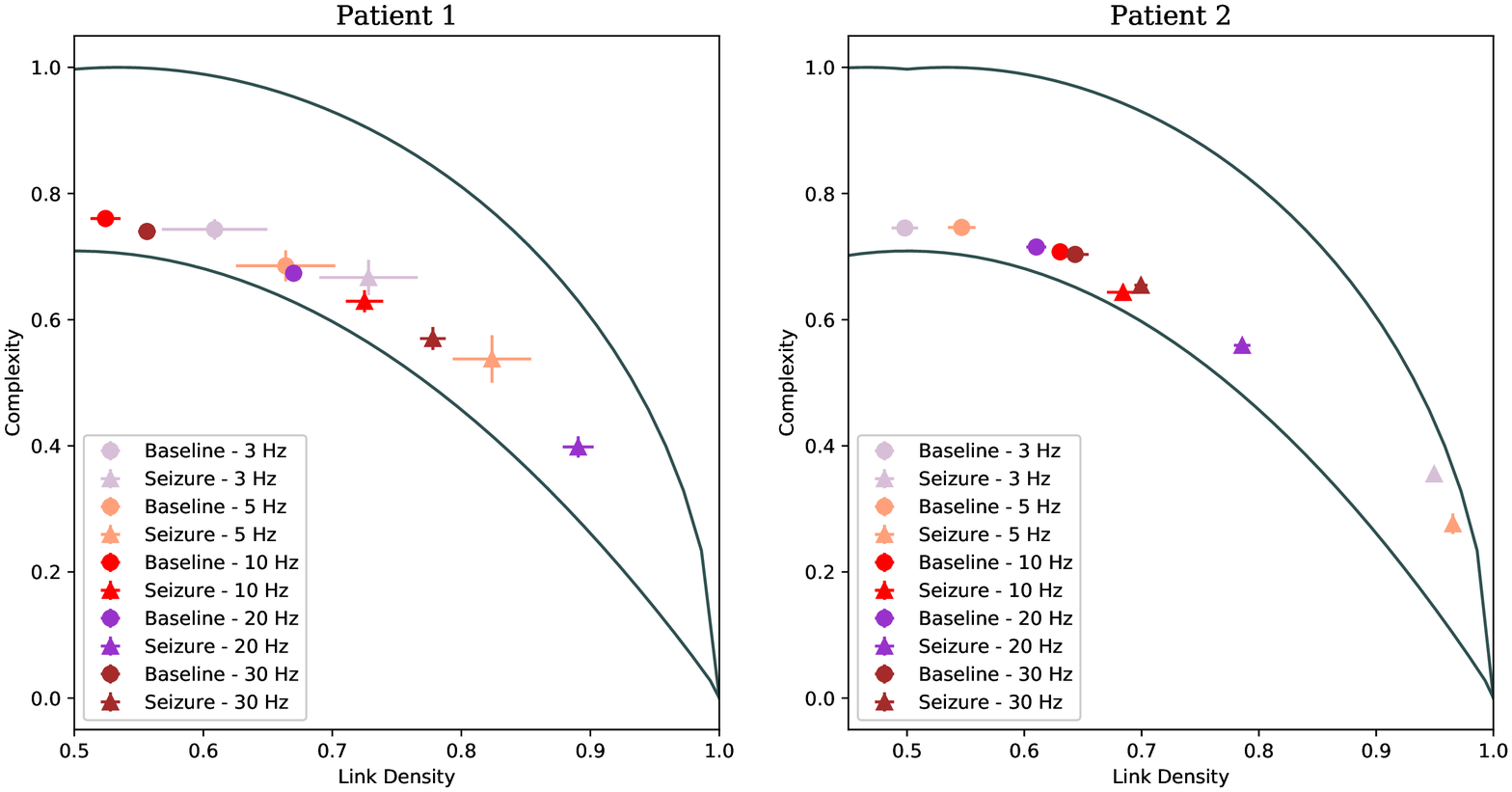}}
\caption{  Complexity analysis for two patients suffering two kinds of epilepsy (Patient 1  primary generalized epilepsy and Patient 2 secondary generalized epilepsy). The circle represent Baseline state and the triangle the Seizure state. Each analysis was performed for different central frequencies $[3, 5, 10, 20, 30] Hz$. Each point/bar represents the mean value/standard deviation performed over 15 MEG signal epochs. 
  }
\label{fig:CompSz}
\end{figure}

Figure \ref{fig:CompSz} shows the complexity analysis for patient 1 (left) and patient 2 (right). The circles represent the values for the baseline state, while the triangles are seizure state. Each point is the mean value calculated over 17 matrices and the bars represent their respective error. 

Patient 1 shows a higher complexity and intermediate link density for the baseline state while for the seizure state, less complexity and greater connectivity are observed for the 5 frequencies analysed. 
The main difference between states is observed for high frequencies (20 and 30 Hz).  Similar behaviour shows patient 2, however the most marked differences between states occur for the low frequencies (3 and 5 Hz).

The result shows that complexity can discern between both states for all frequencies. In the Seizure state, the system generates a high global synchronisation of the brain. This makes the analysis results in highly connected connectivity matrices (high link density). Furthermore, it can be observed that these matrices have a low complexity. Low complexities for Seizure states have been found in other studies \cite{mateos2017consciousness,erra2016statistical}.

The baseline state presents a high complexity because the connectivity in the graph is more distributed among its nodes. This graph topology allows the brain to integrate and segregate information more efficiently. These results it are in agreements with the  Information Integration Theory Hypothesis, which claims that in awake state, the integration and segregation of information in the brain tends to be maximal \cite{tononi2004information}.

\section{Discussion}
\label{sec:Discussion}

Our aim in this note is to introduce an energy based notion of complexity for graphs. As in classical quantum mechanics, the energy of a ``free particle'' 
is provided by the best known Hamiltonian, to wit, the Laplacian. Moreover the admissible quantified  values of the energy are reflected in the spectra of the Laplace operator of the setting. The Laplace operator and its spectral theory is well known on undirected graphs. Once the spectral analysis of our graph is carried over, in particular, once the sequence of eigenvalues is obtained, based on the naive 
idea that fully connected and fully disconnected graphs are not at all complex, we define our notion of complexity. Our first experiment with Erd\"os-R\'enyi random graphs taking the link probability $p$ of the model as the independent variable, lead us to the expected ``inverse U'' shaped curve in the plane $p$ versus complexity. Of course such a behavior of the shape of the curve was also observed in the plane of the variables link density--complexity. Moreover the same basic shape 
is observed for the Watt-Strogatz random model. As it should be expected 
for the Barab\'asi--Albert random model, a turn back 
of the curve in the plane link density--complexity is clearly observed. But, more important, this return to the origin of the curve takes place with higher complexity in the sense of our definition. With these basic observations 
at hand we empirically detected, through the analysis of a variety of deterministic graphs, the plane region in the link density-complexity variables spanned by all graphs with any number of nodes. We found that, generically, for each link density we have a diversity of complexities. Moreover, we found that for a fixed link density the less complex are lattice-like graphs and the more complex are star-like graphs. Since our original interest in defining such an energy based complexity for graphs is motivated by brain-connectivity graphs 
related to epilepsy, we include the localization of some of these graphs in our ``croissant shaped'' region in order to wit their complexity differentiation. 




\providecommand{\bysame}{\leavevmode\hbox to3em{\hrulefill}\thinspace}
\providecommand{\MR}{\relax\ifhmode\unskip\space\fi MR }
\providecommand{\MRhref}[2]{%
	\href{http://www.ams.org/mathscinet-getitem?mr=#1}{#2}
}
\providecommand{\href}[2]{#2}

\appendix  
\section{Appendix}
\subsection{Some mathematical results}
\label{sec:Appendix}

The trace of the Laplacian is a feature of interest for the kind of analysis we are conducting. 
\begin{property}
    $\Big |\sum\limits_{i=1}^n \lambda_i\Big | = |\text{trace}(L)| = 2\,\sharp(\E) = 2m$.
\end{property}
The operator $L$ as defined in \eqref{eq:laplacian} is an unnormalized version of the Laplacian. Some of its normalizations lead to it having trace equal to $\sharp(\V)=n$ or to 1.

\begin{proposition}
	Let $G$ be a graph on $n$ nodes with spectrum $\overline{\lambda}$. Let $G^\complement$ be the complement of $G$, with spectrum $\overline{\lambda^\complement}$. Let $Z$ and $F$ be the null and the complete graph, respectively, on $n$ nodes.
	\begin{enumerate}[$(a)$.]
		\item The spectrum of $F$ is given by $\ \overline{\lambda}_F=-n\,\bar{1}= -\overline{n}$, where $\bar{1}$ is the vector of ones in $\mathbb{R}^{n-1}$.
		\item The spectrum of $Z$ is given by $\ \overline{\lambda}_Z=\bar{0}.$
		\item The spectrum of $G^\complement$ is given by $\overline{\lambda^\complement}= -\overline{n} - \overline{\lambda}$.
		\item $C(G) = \left\|\overline{\lambda}\right\| \big\|\overline{\lambda^\complement}\big\| = \left\|\overline{\lambda}\right\| \big\|\overline{n}+\overline{\lambda}\big\|$.
	\end{enumerate}
\end{proposition}
\begin{proof}
    We will prove only the items $(c)$ and $(d)$, the first two are well known and easy to verify. To see $(c)$, notice that the adjacency matrix of $G^\complement$ is given by $W^\complement = \overline{\overline{1}}-I-W$, where $\overline{\overline{1}}$ is the $n\times n$ matrix of all ones and $I$ is the identity matrix. The degree matrix of $G^\complement$ is given by $D^\complement = (n-1)I-D$. So the Laplacian of $G^\complement$ is given by $L^\complement = W^\complement - D^\complement = \overline{\overline{1}}-I-W-((n-1)I-D) = (\overline{\overline{1}}-nI)-(W-D) = 
    L_F - L_G$, where $L_F$ is the Laplacian of the complete graph and $L_G=L$. Observe that the eigenspace associated with the first eigenvalue $\lambda_1=0$ is shared by all the graphs of $n$ nodes, and is the one generated by the eigenvector $\psi_1=\bar{1}$. The complete graph Laplacian has $n$ as eigenvalue with multiplicity $n-1$, and its corresponding eigenspace is the orthogonal complement of $\psi_1$, so we can take $\{\psi_2,\dots,\psi_n\}$ as its orthonormal basis. Hence $\{\psi_1,\psi_2,\dots,\psi_n\}$ constitutes an orthonormal basis of eigenvectors for both Laplacians $L_F$ and $L_G$, and so it does for $L^\complement$. Then $(c)$ is proved and $(d)$ follows. 
\end{proof}

\end{document}